\documentclass[12pt,a4paper]{article}
\usepackage[english]{babel}
\usepackage[latin1]{inputenc}
\usepackage{amsmath}
\usepackage{amsfonts,latexsym}
\usepackage{amsthm}
\usepackage{verbatim}
\usepackage{amssymb}
\usepackage{graphicx}
\usepackage{color}

\newcommand{\be}{\begin{equation}}
\newcommand{\ee}{\end{equation}}
\newcommand{\ba}{\begin{array}}
\newcommand{\ea}{\end{array}}
\newcommand{\bea}{\begin{eqnarray}}
\newcommand{\eea}{\end{eqnarray}}
\newcommand{\bean}{\begin{eqnarray*}}
\newcommand{\eean}{\end{eqnarray*}}

\newcommand\wb{\bar{w}}

\newcommand\Lm{L_+}

\newtheorem{theorem}{Theorem}[section]
\newtheorem{remark}[theorem]{Remark}
\newtheorem{corollary}[theorem]{Corollary}
\newtheorem{lemma}[theorem]{Lemma}

\numberwithin{equation}{section}

\newcommand{\Sch}{Schr\"odinger\,}
\newcommand{\Schp}{Schr\"odinger--Poisson\,}
\newcommand{\R}{\mathbb{R}}
\newcommand{\N}{\mathbb{N}}
\newcommand{\mD}{\mathcal{D}}
\newcommand{\8}{\infty}
\newcommand{\mcH}{\mathcal{H}}
\newcommand{\mcHp}{\mathcal{H'}}
\newcommand{\vi}{\varphi}
\newcommand{\Ai}{\mbox{Ai}}

\DeclareMathOperator{\sgn}{sgn}

\begin{document}

\markboth{M. De Leo, C. S\'anchez Fern\'andez de la Vega, D.
Rial}{Controllability of \Sch  equation with a nonlocal term.}

\title{Controllability of \Sch  equation with a nonlocal term.}

\author{
Mariano De Leo\thanks{Instituto de Ciencias, Universidad Nacional
de General Sarmiento, J.M. Guti\'errez 1150 (1613) Los Polvorines,
Buenos Aires,
Argentina. Email: {\em mdeleo@ungs.edu.ar}}\,,\;%
Constanza S\'anchez Fern\'andez de la Vega\thanks{Departamento de
Matem\'atica, Facultad de Ciencias Exactas y Naturales,
Universidad de Buenos Aires, Ciudad Universitaria, Pabell\'on I
(1428) Buenos Aires, Argentina. Email: {\em csfvega@dm.uba.ar}}\,,\;%
Diego Rial\thanks{Departamento de Matem\'atica, Facultad de
Ciencias Exactas y Naturales, Universidad de Buenos Aires, Ciudad
Universitaria, Pabell\'on I (1428) Buenos Aires, Argentina. Email:
{\em drial@dm.uba.ar}}%
}
\date{march 14, 2012}

\maketitle

\begin{abstract} This paper is concerned with the internal distributed
control problem for the 1D \Sch equation,
$i\,u_t(x,t)=-u_{xx}+\alpha(x)\,u+m(u)\,u,$ 
that arises in quantum semiconductor models. Here $m(u)$ is a non
local Hartree--type nonlinearity stemming from the coupling with
the 1D Poisson equation, and $\alpha(x)$ is a regular function
with linear growth at infinity, including constant electric
fields. By means of both the Hilbert Uniqueness Method and the
Schauder's fixed point theorem it is shown that for initial and
target states belonging to a suitable small neighborhood of the
origin, and for distributed controls supported outside of a fixed
compact interval, the model equation is controllable. Moreover, it
is shown that, for distributed controls with compact support, the
exact controllability problem is not possible.

\vspace{1cm}

Keywords: Nonlinear \Schp\!; Hartree potential; constant electric
field; internal controllability.

\bigskip

AMS Subject Classification: 93B05, 81Q93, 35Q55. 

\end{abstract}


\section{Introduction}

We are mainly concerned with the internal distributed
controllability for the following 1D \Sch equation
\begin{align}\label{int_eq}
&iu_t=-u_{xx}+\alpha(x)\, u+m(u) u, \qquad x\in \R,\; t>0,\\
&u(x,0)=u_0(x) \label{int_in},
\end{align}
posed in the Sobolev space $\mcH=\{\phi\in H^1(\R): \int \mu(x)
|\phi|^2 <\8\},$ where $\mu$ is a positive regular function that
coincides with $|x|$ away from the origin. Here, the non linearity
$m(u)$ is of non local nature:
\begin{equation}\label{defino_m}
  m(\phi)(x)=\int \varrho(x,y) |\phi(y)|^2 dy,
\end{equation}
where the kernel satisfies the estimate $|\varrho(x,y)|\leq
\mu(y).$ This choice is motivated for the self--consistent 1D
\Schp equation used in quantum semiconductor theory where the
Hartree term $u \left(|x|\ast \left(\mD-|u|^2\right)\right),$
after a suitable splitting, reads
\begin{equation*}
a \mu(x) u+u \int \left(|x-y|-\mu(x)\right)
\left(\mD(y)-|\phi(y)|^2\right) dy,
\end{equation*}
where $a\in \R$ is a constant
depending on the size of the initial datum, and $\mD(x)$ denotes
the fixed positively charged background or {\em impurities}, see \cite{MRS} and references
therein for semiconductor models. We
note that in the 1D case the kernel $\mu(x)$ is not bounded nor
integrable so the classic theory developed in \cite{C} does not
apply and we refer \cite{DLR} for details on the well posedness.
In this article we will consider a slightly extended version in
which the term $a\mu(x)$ is replaced by a regular function $\alpha(x)\in
C^\infty(\R),$ with at most linear growth at infinity (i.e. with the asymptotics $\alpha(x) \sim C^\pm x$ for $
x\sim \pm \infty$), in order to include constant electric fields
$\alpha(x)=q x$. We note that due to the regularity requirements of the unique
continuation technique displayed in Lemma \ref{obsHp}, the regular function $\alpha(x)$ appears as a regularized approximation of a locally
constant electric field, which is modelled with a polygonal
function. It is also worth to mention that since the
impurities give rise to a bounded potential
\begin{equation*}
V_d(x)=\int\left(|x-y|-|x|\right) \mD(y) dy,
\end{equation*}
and hence enters in the model
equation as a bounded multiplication operator, and since our
results are still valid for bounded perturbations, there is no loss of
generality in restricting ourselves to the case $\mD\equiv 0.$
Let us finally mention that results on controllability with local
nonlinearities as $|u|^{2\sigma} u$ are widely developed, see
\cite{ILT,RZ}, and therefore will not be taken under
consideration.

The problem of exact internal controllability of equation
(\ref{int_eq})-(\ref{int_in}) is usually described as the question
of finding a control function $h \in L^2(0,T,\mcH)$ and its
associated state function $u\in C(0,T,\mcH)$ such that
\begin{align}\label{exact_cont}
&iu_t=-u_{xx}+\alpha(x)\, u+m(u) \,u  +\psi(x) h(x,t), \hspace{0.2cm} x\in \R,\,\, t\in (0,T),\\
&u(x,t_0)=u_0(x), \hspace{0.5cm} u(x,T)=u_T(x) \label{exact_ini}
\end{align}
where $T>0$ is a given target time and $u_0$ and $u_T$ are the
given initial and target states respectively, and $\psi:\R \to \R$
is a given $C^1$ function that localizes the control to
$\mathrm{Supp}(\psi)$. The problem of distributed
controllability for \Sch equations of nonlinear type appears often
in nonlinear optics, see for instance \cite{MDF,HOBKNT}. There are several
results on controllability of the \Sch equation, for
a review on this topic we refer \cite{EZ}.

In this paper we discuss the internal distributed controllability
for the problem
\begin{align*}
&iu_t=-u_{xx}+\alpha(x)\, u+m(u)\, u,  \quad x\in \R,\; t>0,\\
&u(x,t_0)=u_0(x)
\end{align*}
and present results concerning two different situations depending
on the support of the control: on one hand controls that are
supported outside a compact interval, in which case we shall give
positive results, and on the other hand localized controls, in
which case we shall give a non controllability result.

We start dealing with a distributed control given by $\psi(x)
h(x,t)$ where $\psi\in C^1(\R)$ satisfies:
\begin{equation}\label{phi} \psi(x)=\left\{
\begin{array}{cl}1& \text{for } |x|\geq R+1 \\
0 & \text{for } |x|\leq R
\end{array} \right.
\end{equation}
We thus show that for a given $0<T$ there exist a (small) constant
$\delta$ such that for every $u_0, \; u_T\in \mcH$ with
$\|u_0\|_{\mcH}, \|u_T\|_{\mcH}<\delta$ there exists a control
$h(x,t) \in L^2(0,T,\mcH)$ such that the nonlinear problem
(\ref{exact_cont})-(\ref{exact_ini}) has a unique solution $u\in
C(0,T,\mcH).$

We then turn to the case in which $\psi \in C^1$ is compactly
supported and show that for both $\alpha=\mu$ (linear operator
with a discrete spectrum) and $\alpha(x)=x$ (constant electric
field, which has a continuous spectrum), the linear system is not
exactly controllable. More precisely we show that for any fixed
finite time $T>0$ and any fixed target state $u_T \in \mcH$ there
exist an open bounded interval $\Omega$ and an initial state
$u_0\in \mcH,$ such that for any $\psi$ with
$\mathrm{Supp}(\psi)\subset \Omega,$ there is no control function
$\psi(x) h(x,y),$ with $h\in L^2(0,T,\mcH),$ and no constant
$C=C(T,\Omega)$ such that
\begin{align*}
&iu_t=-u_{xx}+\alpha(x)\, u+\psi(x) h, \qquad x\in \R,\; t\in [0,T],\\
&u(x,0)=u_0(x),\qquad u(x,T)=u_T(x)
\end{align*}
with $\|h\|_{L^1(0,T,L^2(\Omega))}\leq
C(T,\Omega)\left(\|u_0\|_{\mcH}+\|u_T\|_{\mcH} \right).$

The paper is organized as follows. We set the problem in section
1. In section 2, we deal with the existence of dynamics and
establish useful estimates for the related evolution. Section 3 is
devoted to the problem in which the control vanishes inside an
open bounded interval, we start studying the linear system for which we
prove global controllability in the space $\mcH$; we then prove
the local controllability for the nonlinear system
(\ref{exact_cont}). In Section 4, we deal with the non
controllability result for compactly supported controls.

\section{Preliminaries}
In this section we shall collect some results concerning spectral
properties for the operator $-\partial_x^2+\alpha(x).$ Since most
of the estimates refer to different functional spaces we list them
below:

\begin{itemize}
\item $H^1(\R):=\{\phi \in L^2(\R): \phi_x \in L^2(\R)\}.$

\item $L^2_\mu(\R):=\{\phi: \mu^{1/2}\phi \in L^2(\R)\}$ where
$\mu$ is a regular even function satisfying $1\leq \mu(x),$ and
$\mu(x)\equiv |x|$ for $|x|\geq 2.$

\item $\mcH:=H^1(\R)\cap L^2_\mu(\R)$ with
$\|\phi\|_\mcH^2=\|\phi_x\|^2_{L^2}+\|\phi\|^2_{L^2_\mu}$

\end{itemize}

\subsection{Existence of dynamics}
To start with we consider the auxiliar operator $L_+$ defined by
\begin{eqnarray} \label{defino_L+} \nonumber
L_+:\mcH&\mapsto& \mcH'\\
\phi&\mapsto &L_+(\phi):=\left(-\partial_x^2+|x|\right)\phi
\end{eqnarray}

Although this operator does not enter directly in our model,
because of the loss of regularity of $|x|$ in the origin, it
provides the workspace $\mcH$ and also it possesses useful
spectral properties that are easily deduced from the ones of the
Airy function.

\begin{lemma} \label{L+operator}
 The operator $L_+$ verifies:
\begin{itemize}
 \item[(a)] Is self--adjoint.
 \item[(b)] Has a discrete spectrum $0<\lambda_1<\cdots<\lambda_N \nearrow +\infty.$
 \item[(c)] Has a countable set of orthonormal (with respect to $L^2$) eigenfunctions $\{\vi_N: N\in \N\}\subseteq \mathcal{S}(\R)$ satisfying
\begin{equation}\label{radiacion}
  \displaystyle\lambda_N^{-1/4}\int_\Omega |(\vi_N)_x |^2\leq C(\Omega),
\end{equation}
where $\Omega$ is an arbitrary bounded interval.

\end{itemize}

\end{lemma}

\begin{remark}
Self--adjointness of $L_+$ and the existence of both a discrete
spectrum, $\{0<\lambda_1<\lambda_2<\cdots\},$ and an orthonormal
basis of eigenfunctions, $\{\vi_N\}_{N\in \N}\subseteq
\mathcal{S}(\R),$ follows directly from \cite{D} where by means of
variational methods it is only shown that $L_+^{-1}$ is a compact
operator. However, the non--controllability result relies on some
special feature of the eigenfunctions, given by claim (c), that
are not considered there and we shall give an alternative proof.
\end{remark}

\begin{proof}
We first notice that the related quadratic form verifies
$\langle\phi;L_+\phi\rangle=\|\phi_x\|^2_{L^2}+\||x|^{1/2}\phi\|^2_{L^2}$
and this is an equivalent norm for $\mcH,$ from where we recover
the self--adjointness of  $L_+$. The operator $L_+$ has an
explicit spectral decomposition expressed in terms of the Airy
function $\Ai,$ defined as the solution of
$-\Ai_{xx}(x)+x\Ai(x)=0$ such that $\Ai(+\infty)=0,$ as
follows. Let $0<z_0<z_1<\cdots \nearrow +\infty,$ and
$0<w_0<w_1<\cdots\nearrow +\infty$ be the zeros of
$\Ai_{x}(-x)$ and $\Ai(-x)$ respectively, and take
$\lambda_{2N}=z_N,$ $\lambda_{2N+1}=w_N,$ and
$\vi_{2N}(x)=c_{2N}\Ai(|x|-\lambda_{2N}),$
$\vi_{2N+1}(x)=c_{2N+1}\sgn(x) \Ai(|x|-\lambda_{2N+1}),$ where
$c_{N}$ is a (bounded) sequence of normalization constants. A
direct computation shows that $L_+(\vi_N)=\lambda_N \vi_N.$ This
gives the spectral decomposition of $L_+.$ Since for $|x|\sim
+\infty$ it happens that $|x|-\lambda_N>0,$ each eigenfunction
$\vi_N$ inherit the decaying properties of the Airy function near
$+\infty$ where it behaves as $e^{-r^{3/2}}.$ Finally, a standard
bootstrapping argument yields the regularity needed to ensure that
$\vi_N\in \mathcal{S}(\R).$

In order to get claim (c) we take profit of the integral expression for the
Airy function and its derivative, with $x=-|x|,$
\begin{align*}
Ai(x)&=(2\pi)^{-1/2}|x|^{1/2}\,\int  e^{ i|x|^{3/2}(k^3/3-k)} dk\\
\frac{d}{dx}Ai(x)&=(2\pi)^{-1/2}\,\int i k e^{i|x|^{3/2}(k^3/3-k)} dk
\end{align*}
from where, by means of the stationary phase method, we deduce the asymptotics 
\begin{equation}\label{Airy_asym}
\left|\frac{d}{dx}Ai(x)\right|\leq C(M) |x|^{1/4}  
\end{equation}
valid for $x\leq -M,$ and also an estimate for the eigenvalues
\begin{equation}\label{asym}
 \lambda_N\sim N^{2/3}.
\end{equation}

Let $M$ be such that
$\Omega\subseteq [-M,M],$ from estimate \eqref{asym}  there exists
$N_0$ such that, for $N>N_0,$ $\lambda_{2N}-\lambda_N>M.$ Then,
for $x\in \Omega$ one has
$|x|-\lambda_{2N}<M-\lambda_{2N}<-\lambda_N.$ Using
\eqref{Airy_asym} we conclude $|\vi_{2N}(x)|<\lambda_N^{1/4},$
and therefore $\|(\vi_N)_x\|_{L^2(\Omega)}\leq (2M)^{1/2}\lambda_{N/2}^{1/4}.$ This finishes the proof.
\end{proof}

In order to develop the observability inequality we need to build
some appropriate Sobolev spaces, related to the operator $L_+$
defined by \eqref{defino_L+}. This is done as follows. Let
$\{\vi_N\}_{N\in \N}$ be the orthonormal basis of $L^2$ given by
Lemma \ref{L+operator} and, for $\phi \in L^2,$ let
$\widehat{\phi}$ be the Fourier coefficient:
$\widehat{\phi}(N):=\langle\phi;\vi_N\rangle.$ We then set, for
$k=0,1,2$ the Hilbert spaces $W^k:=\{\phi \in L^2: \sum_{N\geq 0}
\lambda_N^{k}\widehat{\phi}(N)^2<\infty\},$ with the inner product
\begin{align}\label{inner_dot}
\langle \psi;\phi\rangle_{W^k}:=\sum_{N\geq 0} \lambda_N^{k}
\widehat{\psi}(N)^* \widehat{\phi}(N).  
\end{align}

Let ${\cal F} \subset W^0$
be the set of finite sums of $\{\vi_N\}_{N \in \N}$. Then for
$k=-3,-2,-1$ the inner product \eqref{inner_dot} is well defined. We then set $W^{k}$ as the
Hilbert space obtained from the closure of $\cal F$ with the norm
induced by $\langle \cdot;\cdot\rangle_{W^k}$. We have that $L_+:W^k \to W^{k-2}$ is an isometry: $\|L_+
w\|_{W^{k-2}}=\|w\|_{W^k}$. Being $L_+$ positive, we have $\Lm^{1/2}:W^k
\to W^{k-1}$ which is also an isometry: $\|\Lm^{1/2}
w\|_{W^{k-1}}=\|w\|_{W^k}$.

We finally mention that $W^0=L^2,$ $W^1=\mcH,$ $W^2=D(L_+),$ the domain of the operator $L_+:W^2\mapsto L^2,$ and $W^{-1}=\mcHp,$ with compact embeddings
\begin{align}\label{compact_inclusion}
  W^{2}\subset W^1 \subset W^0\subset W^{-1} \subset W^{-2}
\end{align}

\begin{remark}
\label{Lmu+} Since $L_{\mu}-L_+=\mu-|x|$ it follows that both $L_{\mu}:W^k \to W^{k-2}$ and $L_{\mu}^{-1}:W^{k-2} \to W^{k}$ are bounded operators.
\end{remark}

We now turn to the general situation $L:=-\partial_x^2
+\alpha(x),$ where $\alpha(x)\in C^\infty(\R)$ is a regular
function verifying $\alpha_{x}, \alpha_{xx}\in
L^\infty,$ and also the asymptotics
\begin{eqnarray}\label{alfa_asintota}
  \displaystyle\lim_{x\to \pm\infty} \frac{\alpha(x)}{\mu(x)}=C^{\pm}
\end{eqnarray}

The following lemma states precisely the self--adjointness result.

\begin{lemma}\label{selfadjoint}
Let $\alpha\in C^\infty(\R)$ satisfying \eqref{alfa_asintota}.
Then $L:\mcH \mapsto \mcHp$ defined by $L:=-\partial_x^2+\alpha(x)$ is
self--adjoint, and therefore generates a strongly continuous group
of unitary operators in $L^2(\R).$
\end{lemma}

\begin{proof}

To this purpose we first show that $L$ is a closed operator. Let
$\vi \in C_0^{\8}(\R)$ and $(\phi_n;L(\phi_n))\in \mcH\times
\mcHp$ a sequence such that $(\phi_n;L(\phi_n))\to
(\phi;\psi)$ in $\mcH\times \mcHp,$ since $\langle
\vi;\phi_{xx}-(\phi_n)_{xx} \rangle=\langle
\vi_{xx};\phi-\phi_n \rangle \to 0$ and $\langle
\vi;\alpha(\phi-\phi_n) \rangle=\langle
\sgn(\alpha)|\alpha|^{1/2}\vi;|\alpha|^{1/2}(\phi-\phi_n) \rangle
\to 0$ we thus have $\langle \vi;L(\phi-\phi_n) \rangle \to 0,$
and consequently we conclude $\langle \vi;\psi-L\phi \rangle =
\langle \vi;\psi-L\phi_n \rangle + \langle \vi;L(\phi_n-\phi)
\rangle  \to 0.$ This shows that $L: \mcH \to \mcHp$ is a
closed operator.

Let now $L_\mu:=-\partial_x^2+\mu(x),$ since $\mu(x)\geq 1$ we deduce that $L_\mu\geq I$ (the identity
operator). For $\vi, \psi \in \mcH$ we set the (well defined)
quadratic form $\mathcal{Q}(\phi,\psi):=\langle \phi_x; \psi_x
\rangle + \langle \phi;\alpha(x) \psi \rangle.$
We now establish two useful estimates

\begin{align*}
  \left|\mathcal{Q}(\phi;\psi)\right|&\leq |\langle \phi_x; \psi_x \rangle|+ |\langle \phi; \alpha\psi \rangle|\\
  & \leq (1+\|\alpha \mu^{-1}\|_{L^\infty})\left|\langle \phi; L_\mu \psi \rangle\right|\\
  &\leq (1+\|\alpha\mu^{-1}\|_{L^\infty})\|L_\mu^{1/2} \phi\|_{L^2}\, \|L_\mu^{1/2}\psi\|_{L^2}\\
  \\
  \left|\mathcal{Q}(L_\mu\phi;\psi)-\mathcal{Q}(\phi;L_\mu \psi)\right|&= \left|\langle \phi; [L_\mu:L] \psi \rangle \right|\\
  &\leq |\langle \phi;(\mu-\alpha)_{xx} \psi \rangle |+
   |\langle (\mu-\alpha)_{x} \phi; \psi_x \rangle |\\
  &\leq \left(\|(\mu-\alpha)_{xx}\|_{L^\infty}+
  \|(\mu-\alpha)_{x}\|_{L^\infty} \right) \|L_\mu^{1/2}\phi \|_{L^2} \, \|L_\mu^{1/2}\psi \|_{L^2}
\end{align*}
where we have used the identity $\|L_\mu^{1/2}\vi
\|_{L^2}^2=\|\vi_x\|_{L^2}^2+\|\vi\|_{L^2_{\mu}}^2.$ Applying Theorem X.36' in
\cite{RS} we obtain that $L$ is a essentially
self--adjoint operator in $\mcH,$ since it is closed, it follows
that $L$ is self adjoint.
\end{proof}

\subsection{Scattering properties for constant electric fields}

The non controllability result, see Theorem \ref{no_controla}, for
a constant electric field $L_e:=-\partial_x^2 -x,$ follows from a
well-known $L^1-L^\infty$ estimate for the group $U_e(t)$
generated by $L_e,$ which depends upon a result of Avron-Herbst,
see \cite{S} for details.

\begin{lemma}
The operator $L_e$ is essentially self--adjoint on
$\mathcal{S}(\R)$
 and
\begin{equation} \label{electric_group}
 U_e(t)=e^{-it^3} e^{itx} e^{-i(p^2t+t^2p)}
\end{equation}
where $p=-i\partial_x$ is the momentum operator.

\end{lemma}

\begin{remark}
Identity \eqref{electric_group} says that except for phase factors
$U_e(t)\phi(x)$ is obtained by first translating by $t^2$ units to
the right and then applying the free particle group
$e^{it\partial_x^2}$
\end{remark}

\begin{corollary}  \label{L1_L8_estimate}
 For $\phi \in L^1(\R)$ we have the following estimate:
\begin{equation*}
 \|U_e(t)\phi\|_{L^\infty}\leq  C |t|^{-1/2} \|\phi \|_{L^1}
\end{equation*}

\end{corollary}

\subsection{Estimates for the evolution}

Lemma \ref{selfadjoint} guarantees that $L$  generates a group
$U(t)$. In the sequel we will exhibit useful bounds for the
evolution related to both the homogeneous and inhomogeneous
problem.

\begin{lemma}\label{semigroup-estimate}
Let $U(t)$ be the group generated by $L:=-\partial^2_x+\alpha$ in $\mcH.$ Then
\begin{itemize}
\item $\|(U(t)\phi)_x\|_{L^2}\leq \|\phi_x\|_{L^2}+|t| \|\alpha_{x}\|_{L^\infty}\|\phi\|_{L^2}.$

\item $\|U(t)\phi\|_{L^2_\mu}\leq \|\phi\|_{L^2_\mu}+ 2^{1/2} |t|^{1/2} \|\phi\|^{1/2}_{L^2} \|\phi_x\|^{1/2}_{L^2}
+|t| \|\alpha_{x}\|_{L^\infty}\|\phi\|_{L^2}.$

\item $\|U(t)\phi\|_{\mcH}\leq \|\phi\|_{\mcH}\big( 1+|t|\cdot
\|\mu_{x}-\alpha_{x}\|_{L^\infty}\big).$
\end{itemize}
\end{lemma}

\begin{proof}
Let $u(t)=U(t)\phi,$ since $u$ verifies $iu_t=-u_{xx}+\alpha u$ and $\|u\|^2_{\mcH}=\|u\|_{L^2_\mu}^2+\|u_x\|_{L^2}^2$ we have
\begin{align*}
 \frac{d}{dt}\big\langle u_x;u_x \big\rangle_{L^2}&=2\mbox{Re}\big\langle u_{xt};u_x\big\rangle_{L^2}\\
 &=2\mbox{Re}\big\langle u_x;-i\alpha_{x} u
\big\rangle_{L^2}\\
 \frac{d}{dt}\big\langle u;\mu u \big\rangle_{L^2}&=2\mbox{Re}\big\langle u_{t};\mu u\big\rangle_{L^2}\\
 &=2\mbox{Re}\big\langle u_x;i\mu_{x} u
\big\rangle_{L^2}\\
 \frac{d}{dt}\big\langle u; u \big\rangle_{\mcH}&=2\mbox{Re}\big\langle u_x;i(\mu_{x}-\alpha_{x}) u
\end{align*}
from where, using a standard ODE argument we conclude the required inequalities.
\end{proof}

We now turn our attention to the non linear term in equation
(\ref{exact_cont}), and give the following estimates.

\begin{lemma} \label{estimate_non_linear}
Let $m:\mcH\mapsto L^\infty(\R)$ be given by
\begin{equation*}
  m(\phi)(x)=\int \varrho(x,y) |\phi(y)|^2 dy.
\end{equation*}
where $|\varrho(x,y)|\leq \mu(y)$ and $\left|\varrho_x(x,y)\right|\leq C$. Then for $\phi, \phi_1 \in \mcH$ the following estimates hold.

\begin{itemize}
  \item $\|m(\phi)\|_{L^\8}\leq \|\phi\|_{L_\mu^2}^2$

  \item $\|m(\phi) \phi-m(\phi_1) \phi_1\|_{\mcH}\leq 3/2
  \left( \|\phi\|_{\mcH}^2 + \| \phi\|_{\mcH} \| \phi_1\|_{\mcH}+ \| \phi_1\|_{\mcH}^2\right) \|\phi-\phi_1 \|_{\mcH}$

\end{itemize}

\end{lemma}

\begin{proof}

It is a straightforward computation and will be omitted.

\end{proof}

We now turn to the non homogeneous problem \eqref{exact_cont} and
give similar estimates in the lemma below, which in turn
express the global well posedness of the problem.

\begin{lemma} \label{evol_norm}
Let $T>0$ be fixed, and let $u\in C(0,T,\mcH)\cap
C^1(0,T,\mcHp)$ be a solution of \eqref{exact_cont} with
fixed $h\in L^2(0,T,\mcH)$ and $\psi\in C^1(\R)$ such that $\psi,$
and $\psi_{x} \in L^\infty(\R).$ Then we have the following
estimates:

\begin{itemize}
\item $\|u\|_{L^\infty(0,T,L^2(\R))}\leq \|u_0\|_{L^2(\R)}+T^{1/2}
\|\psi\|_{L^\infty} \|h\|_{L^2(0,T,\mcH)}.$

\item $\|u_x\|_{L^\infty(0,T,L^2(\R))}\leq \|u_0\|_{H^1}+
\|h\|_{L^2(0,T,\mcH)}\, T^{3/2}\,C(u_0,\psi).$

\item $\|u\|_{L^\infty(0,T,L^2_\mu(\R))}\leq \|u_0\|_{L^2_\mu}+ \|h\|_{L^2(0,T,\mcH)}\,T^{3/2}
\,C(u_0,\psi).$

\item $\|u\|_{L^\infty(0,T,\mcH)}\leq \|u_0\|_{\mcH}+ \|h\|_{L^2(0,T,\mcH)}\,T^{3/2}
\,C(u_0,\psi).$

\end{itemize}
\end{lemma}

\begin{proof}

Relies on a similar procedure as the one displayed in Lemma
\ref{semigroup-estimate} and will be omitted.

\end{proof}

\section{Controllability}

\subsection{Linear system}

We start this section taking into consideration the
controllability of the linear problem, which throughout this
section means the existence of a control $h(x,t)$ such that the
unique solution of the related non homogeneous linear equation
\begin{align} \label{linealhomogeneo}
iu_t(x,t)&=L u(x,t)+\psi(x) h(x,t) \\
u(x,0)&=u_0(x) \quad x\in \mathbb{R} \label{datoinicial}
\end{align}
satisfies $u(x,T)=u_T(x),$ for given $T>0$ and $u_0, u_T(x)\in
\mcH,$ where $L:=\partial_x^2+\alpha(x)$ is the operator of Lemma
\ref{selfadjoint}, and $\psi$ is defined in \eqref{phi}. The main
result is given in the following theorem; its proof is based on
the Hilbert Uniqueness Method (HUM), requires some technicalities,
which we shall first develop, and will be delayed until the end of
this subsection.

\begin{theorem} Global controllability: linear case.\medskip

\label{controlabilidadlineal} Let $T>0$ be given. Then there
exists a bounded linear operator $G: \mcH \times \mcH \to
L^2(0,T,\mcH)$ such that for any $u_0,u_T \in \mcH$  the system
\eqref{linealhomogeneo}-\eqref{datoinicial}, with $h=G(u_0,u_T)$,
admits a solution $u\in C(0,T,\mcH)$ satisfying $u(x,T)=u_T$.
\end{theorem}

\bigskip

As we stated before, we need first to present the ingredients to
apply the HUM. To do this, we consider the corresponding adjoint
problem in $\mcH'$:
\begin{align}
iv_t(x,t)&=L v, \label{eqdual}\\
v(x,0)&=v_0(x) \label{datoinicialdual}.
\end{align}


Let $\Lambda:\mcH \to \mcHp$ denote the usual isomorphism between
the real spaces $\mcH$ and $\mcHp$ defined by $\Lambda(v)=\langle
v,\cdot \rangle_\mcH$. Given $v_0\in \mcHp$, let $v$ be the
solution of equation \eqref{eqdual}. Then, take
$h(\cdot,t)=\Lambda^{-1}(\psi v(\cdot,t))$ and consider the
problem
\begin{equation} \label{eqpatras} \left\{
\begin{array}{cl}
iw_t(x,t)&=L w+\psi(x) h(x,t), \\
w(x,T)&=u_1(x),
\end{array} \right.
\end{equation}
which we split into the two problems:
\begin{equation} \label{eqpatrassinh} \left\{
\begin{array}{cl}
iw^{(1)}_t(x,t)&=L w^{(1)}, \\
w^{(1)}(x,T)&=u_1(x),
\end{array} \right.
\end{equation}
and
\begin{equation} \label{eqpatras0} \left\{
\begin{array}{cl}
iw^{(2)}_t(x,t)&=L w^{(2)}+\psi(x) h(x,t), \\
w^{(2)}(x,T)&=0.
\end{array} \right.
\end{equation}

Clearly, $w=w^{(1)}+w^{(2)}$. As usual with the HUM procedure,
given $v_0\in \mcHp$ the initial condition of equation
\eqref{eqdual}, we define the linear operator $S:\mcHp \to \mcH$
by
\begin{equation}
S(v_0)=-iw^{(2)}(\cdot,0)
\end{equation}
where $w^{(2)}$ is the solution of \eqref{eqpatras0}.

If we can show that $S$ is an isomorphism, then the inverse image
by $S$ of $-iu_0+iw^{(1)}(\cdot,0)$, is the initial condition for
equation \eqref{eqdual} that will provide the sought control
$h=\Lambda^{-1}(\psi v(\cdot,t))$.

This is shown by establishing the observability inequality of
system \eqref{eqdual} in $\cal H'$ which we describe in the
following lemma.

\begin{lemma} \label{obsHp}
Let $\psi$ be a $C^1$ function defined by \eqref{phi}. There
exists a constant $C>0$ such that for all $v_0 \in \mcHp$, the
solution $v$ of \eqref{eqdual}-\eqref{datoinicialdual} satisfies
\begin{equation} \label{observabilityHp}
\int_0^T\|\psi v(.,t) \|^2_\mcHp dt \geq C \|v_0\|^2_\mcHp.
\end{equation}
\end{lemma}

The proof of the observability inequality \eqref{observabilityHp}
is quite similar as the one given by L. Rosier and B. Zhang in
\cite{RZ}. We repeat most of the construction given in that paper
for the sake of completeness.

\vspace{0.5cm}

In order to prove Lemma \ref{obsHp} we begin by proving the
corresponding observability inequality in $\cal H$. We recall the
isomorphism $L_{\mu}:\mcH \to \mcHp$, $L_{\mu}=-\partial_x^2+\mu $. Consider the \Sch equation

\begin{align} \label{eqdualpatras}
iw_t(x,t)&=L w+ P(w), \\
w(x,0)&=w_0(x), \label{eqdualpatrasinicial}
\end{align}
where $P(w)=L_{\mu}^{-1} [\nu,L_{\mu}](w)$, with $\nu:=\alpha-\mu$.

\begin{lemma} \label{P_operator}
 $P:W^k\mapsto W^k$ is a bounded operator for $k=0,1$. Hence
 $\|w\|_{L^\infty(0,T,W^k)}\leq C(T) \|w_0\|_{W^k}.$
\end{lemma}

\begin{proof}
Since $$L_{\mu}^{-1} [\nu,L_{\mu}]=-L_{\mu}^{-1}\nu_{xx}-2L_{\mu}^{-1}\mu_{x}L_{\mu}^{-1}+2 \partial_x L_{\mu}^{-1} \nu_x$$ and
$\mu_x,\nu_x, \nu_{xx}\in L^\infty,$ the proof follows from Remark \ref{Lmu+}.
The estimate is a consequence of Gronwall.
\end{proof}

\begin{lemma} \label{obsH}
Let $\psi$ be a $C^1$ function defined by \eqref{phi}. There
exists a constant $C>0$ such that for every $w_0 \in \mcH$, the
solution $w$ of \eqref{eqdualpatras}-\eqref{eqdualpatrasinicial}
satisfies \be  \label{observabilityH} \int_0^T \|\psi
\,w(\cdot,t)\ \|^2_{\mcH} \,dt  \geq C \|w_0 \|^2_{\mcH}. \ee
\end{lemma}

\begin{proof}
By Duhamel, we know that there exists $C>0$ such that for $w_0
\in \mcH$, the solution $w$ of
\eqref{eqdualpatras}-\eqref{eqdualpatrasinicial} satisfies

\begin{equation} \label{casiconserv}
 \|w_0\|^2_\mcH \leq C \int_0^T \|w(\cdot,t) \|^2_{\mcH} \, dt.
\end{equation}

Therefore, \eqref{observabilityH} will follow if we prove the following inequality in $\mcH$:
\be  \label{observabilitycasiH}
 \int_0^T \|w(\cdot,t) \|^2_{\mcH} \, dt \leq C \int_0^T \|\psi \,w(\cdot,t)\ \|^2_{\mcH} \,dt .
\ee
We use the multiplier technique. Define $q\in C_0^\infty([0,T]\times \R)$
\begin{equation}\label{qu} q(x)=\left\{
\begin{array}{cl}x &\text{for } |x|\leq R+2 \\
0 & \text{for } |x|\geq R+3
\end{array}. \right.
\end{equation}

We have that \be \int_0^T \frac{d}{dt}\langle w,iqw_x \rangle
dt=\langle w,iqw_x \rangle \big|_0^T. \ee Recall
$L=-\partial_x^2+\alpha,$ then the l.h.s of the last equation
reads:
\be
\int_0^T \langle iw_{xx},iqw_x \rangle - \langle i
\alpha w+P(w),iqw_x \rangle +\langle w,iqiw_{xxx} \rangle +
\langle w,iq(-i)(\alpha w+P(w))_x \rangle dt
\ee

Using parts we have that: \be \langle w,iqw_x \rangle
\big|_0^T=\int_0^T -2 \langle w_x,q_{x} w_{x} \rangle -2
\langle \alpha w+P(w),qw_x \rangle - \langle
q_{xx}w,w_{x} \rangle -\langle \alpha w+P(w),q_{x} w
\rangle dt \ee and therefore, using that $\langle f,g\rangle=Re
\int_{\R} f g^*$: \be \frac{1}{2} Im \int_{\R}q w\wb_x
\big|_0^T+Re \int_0^T \int_{\R} \left[q_{x} |w_x|^2
+\frac{1}{2} q_{xx} w \wb_x + (\alpha
w+P(w))(q\wb_x+\frac{1}{2} q_{x}\wb )\right] dx dt=0. \ee Then
\be \ba{rcl} \left|\int_0^T \int_{|x|\leq R+2} |w_x|^2 \right|
&\leq & \frac{1}{2}\left| \int_{\{|x|\leq R+3\}}q w\wb_x \big|_0^T
\right| +
\int_0^T \left[ \left| \int_{\{R+2\leq |x|\leq R+3\}}q_{x} |w_x|^2 \right| \right.\\
&+&\left.\frac{1}{2} \left| \int_{\{R+2\leq |x|\leq R+3\}}
q_{xx} w \wb_x \right|+ \left|\int_{\{|x|\leq R+3\}}
(\alpha w+P(w))(q\wb_x+\frac{1}{2} q_{x}\wb )\right| \right]
\ea \ee and using Lemma \ref{P_operator} and
\begin{align}
&\|w(t_0,.)\|^2_{H^1(\R)} \leq C \int_0^T\|w(t,\cdot)\|^2_{H^1(\R)}dt \,\, \forall t_0\in [0,T] \\
&\|\alpha w\|_{L^2(\{|x|\leq R+3\})}\leq C \|w\|_{L^2(\R)} 
\end{align}
we have that there exist $\varepsilon>0$ and a constant
$C_\varepsilon$ such that
\begin{align} \label{cotawx}
\int_0^T
\int_{|x|\leq R+2} |w_x|^2 dx dt &\leq \varepsilon \int_0^T
\|w(t,\cdot)\|^2_{H^1}dt+ C_\varepsilon \int_0^T
\|w(t,\cdot)\|^2_{L^2}dt \\
&\quad+C_2  \int_0^T \int_{\{R+2\leq |x|\leq R+3\}}|w_x|^2dxdt.
\end{align}

We have that
\begin{equation} \label{cotaw}
\| w\|_{\mcH} \leq \|\psi w\|_{\mcH}+\|(1-\psi) w\|_{\mcH}
\end{equation}
and since $1-\psi=0$ for $|x|>R+1$
\begin{equation}
\| (1-\psi)w\|_{\mcH} \leq C \|(1-\psi) w\|_{H^1}.
\end{equation}
It is clear that
\begin{equation} \label{cota1menosphiw}
 \|(1-\psi) w\|^2_{H^1} \leq C\left( \int_{|x|\leq R+1} |w_x|^2 dx + \|w\|^2_{L^2(\R)} \right),
\end{equation}
and since $(\psi w)_x=w_x$ for $|x|\geq R+2$, we have that
\be \label{cotawxafuera}
\int_{|x|\geq R+2} |w_x|^2 dx \leq \|\psi w \|^2_{\mcH}.
\ee
Therefore, if $\varepsilon$ is chosen small enough, from \eqref{cotawx} and \eqref{cotaw}-\eqref{cotawxafuera}, it follows the inequality
\begin{equation}  \label{cotaintwH}
 \int_0^T \|w(\cdot,t) \|^2_{\mcH} \, dt \leq C\left(\int_0^T \|\psi \,w(\cdot,t)\ \|^2_{\mcH} \,dt + \int_0^T  \|w(\cdot,t) \|^2_{L^2}\, dt\right).
\end{equation}

It remains to prove that
\be \label{remainineq}
\int_0^T \|w(\cdot,t)\|^2_{L^2}dt \leq C \int_0^T \|\psi w(\cdot,t)\|^2_{\mcH} \,dt.
\ee

Assume inequality \eqref{remainineq} is not true, then there exists a sequence ${w_0^k} \in \mcH$ such that the corresponding sequence $w^k$ of solutions of \eqref{eqdualpatras} satisfies
\be \label{negacioncotaL2}
1=\int \|w^k(t)\|^2_{L^2(\R)}dt \geq k \int_0^T \|\psi w^k(t)\|^2_{\mcH}dt, \;\; k=1,2,\dots
\ee
According to \eqref{cotaintwH} and \eqref{negacioncotaL2}, the sequence
$\{w^k\}$ is bounded in $L^2(0,T,\mcH)$.
Therefore by \eqref{casiconserv} the sequence $\{w_0^k\}$ is bounded in $\mcH$. Extracting a subsequence if needed, we may assume that
\be
w_0^k\rightharpoonup w_0 \;\text{ weakly in } \mcH \, \text{ and } \, w^k\rightharpoonup w \;\text{ weakly in } L^2(0,T;\mcH)
\ee
where $w \in C([0,T];\mcH)$ solves equation \eqref{eqdualpatras}-\eqref{eqdualpatrasinicial} with initial data $w_0$.
Indeed, we first have that $w_0^k \rightharpoonup w_0$ weakly in $\mcH$ and $w^k \rightharpoonup u$ weakly in $L^2(0,T,\mcH)$.
Being $\mcH$ compactly imbedded in $L^2(\mathbb{R})$, we may assume that $w_0^k \to w_0$ strongly in $L^2(\mathbb{R})$
and therefore
\begin{equation} \label{convfuertewk}
w^k \to w \text{ strongly in }  L^2(0,T,L^2(\mathbb{R}))
\end{equation}
where $w\in C(0,T,\mcH)$ since is the solution of
equation \eqref{eqdualpatras}-\eqref{eqdualpatrasinicial} with initial data $w_0\in \mcH$.
From the uniqueness of weak limit in $L^2(0,T,L^2(\mathbb{R}))$ we obtain that $w=u$.

By \eqref{negacioncotaL2}, $\psi w^k\to 0$ strongly in $L^2(0,T,\mcH)$ and since
$\psi w^k\rightharpoonup  \psi w$ weakly in $L^2(0,T,\mcH)$, we conclude that $\psi w\equiv 0$
on $\R \times (0,T)$.
Consequently,
\be
w(x,t)=0, \, |x|>R+1, \, t\in (0,T).
\ee
Let $v=L_+w$, then $v$ satisfies equation \eqref{eqdual} and
\be \label{ceroR+1}
v(x,t)=0, \, |x|>R+1, \, t\in (0,T).
\ee
We consider the new problem (similar to \eqref{eqdual})
\begin{eqnarray} \label{eqdualacotado}
iv_t&=&-v_{xx}+\alpha \psi v \\
v(x,0)&=&v_0. \nonumber
\end{eqnarray}
where $\psi$ is a $C_0^\infty(\R)$ given by
\begin{equation}\label{psi} \psi(x)=\left\{
\begin{array}{cl}1& \text{for } |x|\leq R+1 \\
0 & \text{for } |x|\geq R+2
\end{array}. \right.
\end{equation}
Then, problems \eqref{eqdual} and \eqref{eqdualacotado} have
the same solution which satisfy \eqref{ceroR+1}.
Using Proposition 2.3 from \cite{RZ} with $a=-\alpha\psi$ and $b=0$ functions in $C_0^\infty(\R)$ and
being $v_{0}\in\mcHp$ with compact support, 
we have that $v$ is of class $C^\infty$ on $\R \times (0,T)$.

By the unique continuation property for \Sch equation we conclude
that $v\equiv 0$ on $\R \times (0,T)$. This implies $w\equiv 0$ on
$\R \times (0,T)$. From \eqref{convfuertewk} and
\eqref{negacioncotaL2} we have a contradiction.

Then observability inequality in $\mcH$ \eqref{observabilityH} is proved.
\end{proof}

We are now in position to prove the observability inequality \eqref{observabilityHp} in $\mcHp$.
We first prove a weak inequality:

\begin{lemma} There exists a constant $C>0$ such that for every $v_0\in \mcHp=W^{-1}$ and $v$ the solution
of equation \eqref{eqdual}-\eqref{datoinicialdual}, the following inequality is satisfied
\be \label{weakcotaHp2}
\|v_0 \|_{W^{-1}}^2 \leq C\left(\int_0^T \|\psi v(t)\|^2_{W^{-1}}dt+\|v_0\|^2_{W^{-2}} \right).
\ee
\end{lemma}
\begin{proof}
Suppose that inequality \eqref{weakcotaHp2} is false.  Then there exist a sequence $v_k$ of solutions of \eqref{eqdual} in $C(0,T,\mcH')$
such that
\be \label{negacionweakcota2}
1=\|v_k(0)\|^2_{W^{-1}} \geq k\left(\int_0^T \|\psi v_k(t)\|^2_{W^{-1}}dt+\|v_k(0)\|^2_{W^{-2}} \right).
\ee
Then we can extract a subsequence such that $v_k(0) \to v_0$ weak in $\mcH'$ for some $v_0 \in \mcHp$ and we can assume $v_k \to 0$ strongly en $W^{-2}$ and therefore $v_0=0$. Moreover, we have can assume $\psi v_k \to 0$ strongly in $L^2(0,T,\mcHp)$.

We prove now that if $w\in \mcH=W^{1}$, then  $w_x\in L^2=W^0$
\be \label{deriv1}
\|w_x\|_{W^0} \leq \|w\|_{W^1}
\ee
Let $w\in W^{1}$, then
\be \|w_x\|_{L^2}^2=\langle w_x,w_x \rangle=-\langle w_{xx},w \rangle \leq \langle -w_{xx},w \rangle
+\langle \mu w,w \rangle=\langle L_+ w,w \rangle=\|L_+^{1/2} w\|^2=\|w\|_{W^{1}}^2 \ee

Now, let $v\in \mcHp=W^{-1}$, there exists $w\in \mcH=W^{1}$ such that $v=L_{+}w$, then
\be
\label{deriv2}
\|v_x\|_{W^{-2}}=\|L_{+}w_x+\mu_x w\|_{{W^{-2}}}=\|L_{+}^{-1}(L_{+}w_x+\mu_x w)\|_{W^{0}}\leq\|w_x\|_{W^{0}}+\|L_{+}^{-1}\mu_x w\|_{W^{0}}
\ee
using (\ref{deriv1}), we have $\|v_x\|_{W^{-2}}\leq C\|v\|_{W^{-1}}$.

From (\ref{deriv1}) and (\ref{deriv2}), using interpolation in Hilbert spaces, we get that there exists a constant $C>0$ such that for all $w\in L^2=W^0$
\be
\label{deriv3/2}
 \|w_x\|_{W^{-1}} \leq C\|w\|_{W^{0}}.
\ee

Next, we will prove that ${v}_k(0) \to 0$ strongly in $W^{-1}$ arriving to a contradiction by \eqref{negacionweakcota2}.

Let $w_k=L_+^{-1}(v_k)$, then $w_k \in C([0,T],W^{1})$ is a
solution of equation \eqref{eqdualpatras} in $\mcH$ and \be \psi
w_k=\psi \Lm^{-1} v_k=\Lm^{-1} (\psi
v_k)+[\psi,\Lm^{-1}]v_k=\Lm^{-1} (\psi v_k)+\Lm^{-1}[\Lm,\psi]w_k.
\ee

Since $\psi v_k \to 0$ strongly in $L^2(0,T,\mcHp)$ and
$\|\Lm^{-1} (\psi v_k)\|_{1}=\|\psi v_k\|_{-1}$, we deduce that
$\Lm^{-1} (\psi v_k) \to 0 $ strongly in $L^2(0,T,\mcH)$. On the
other hand, using \eqref{deriv3/2}

\begin{align*}
\|\Lm^{-1} [\Lm,\psi](w_k)\|_{W^{1}}&=\| [\Lm,\psi](w_k)\|_{W^{-1}}\\
&=\| \psi_{xx} w_k+2\psi_x (w_k)_x\|_{W^{-1}} \\
&\leq C (\| v_k\|_{W^{-3}}+\| v_k \|_{W^{-2}})\\
&\leq C \| v_k \|_{W^{-2}}
\end{align*}

implying $\Lm^{-1} [\Lm,\psi](w_k) \to 0$ strongly in
$L^2(0,T,\mcH)$, since $v_k(0) \to 0$ strongly in $W^{-2}$.

Therefore $\psi w_k \to 0$ strongly in $L^2(0,T,\mcH)$. Since
$w_k$ is a solution of \eqref{eqdualpatras} we have from the
observability inequality that $w_k(0) \to 0$ strongly in $\mcH$.
It follows that $v_k(0)=\Lm w_k(0) \to 0 $ strongly in $\mcHp$,
which contradicts the fact that $\|v_k(0)\|_{\mcHp}=1$.
\end{proof}

\begin{proof}[Proof of Lemma \ref{obsHp}]
Assume that inequality \eqref{observabilityHp} is false, then
there exists a sequence $v_k$ of solutions of \eqref{eqdual} in
$C([0,T];\mcHp)$ such that \be \label{nocontrolabilidad}
1=\|v_k(0)\|_{W^{-1}}^2 \geq k \int_0^T \|\psi v_k(t)\|_{W^{-1}}^2 dt \ee
for all $k\geq 0$.

Extracting a subsequence, we may assume that
\be \ba{rl}
&v_k \to v \hspace{0.3cm} \text{in} \,\,\, L^\infty(0,T;\mcHp) \,\,\, \text{weak}-\star, \\
&v_k(0) \to v(0) \hspace{0.3cm} \text{weakly in} \,\,\, \mcHp \ea
\ee for some solution $v\in C(0,T;\mcHp)$ of
\eqref{eqdual}-\eqref{datoinicialdual}. From
\eqref{nocontrolabilidad}, $\psi v_k \to 0$ strongly in
$L^2(0,T,\mcHp)$ and since $\psi v_k\to \psi v$ in
$L^\infty(0,T;\mcHp)$ weak-$\star$, we have that $\psi v\equiv 0$.
We deduce as before that $v\equiv 0$ in $\R \times (0,T)$.

Being $\{v_k(0)\}$ a bounded sequence in $W^{-1}$ and since
$W^{-1}$ is compactly imbedded in $W^{-2},$ see \eqref{compact_inclusion}, there exists a
subsequence such that $v_k(0)$ converges strongly in $W^{-2}$
necessarily to $0$.

We infer from \eqref{weakcotaHp2} that $v_k(0)$ converges strongly
to $0$ in $W^{-1}$ which is absurd from \eqref{nocontrolabilidad}.
This finishes the proof.
\end{proof}

\begin{proof}[Proof of Theorem \ref{controlabilidadlineal}]

Let $h \in L^2(0,T,\mcH)$ and $v_0 \in \mcHp$, and let $w^{(2)}$
be the solution of \eqref{eqpatras0} that verifies
$w^{(2)}(\cdot,T)=0$ and $v(x,t)$ the solution of \eqref{eqdual}
such that $v(x,0)=v_0$. Then

\begin{equation}
\int_0^T \langle v,iu_t-Lu \rangle_{\mcHp,\mcH}dt=\int_0^T \langle v,\psi h \rangle_{\mcHp,\mcH}dt.
\end{equation}
Using that
\begin{equation}
\begin{array}{lcl}
\langle v,iu_t \rangle_{\mcHp,\mcH}&=&\frac{d}{dt}\langle v,iu \rangle_{\mcHp,\mcH}+ \langle i v_t,u \rangle_{\mcHp,\mcH}\\
\langle v,\partial_x^2 u \rangle_{\mcHp,\mcH}&=&\langle \partial_x^2  v,u \rangle_{\mcHp,\mcH} 
\end{array}
\end{equation}
we obtain
\be
\int_0^T \frac{d}{dt}\langle v,iu \rangle_{\mcHp,\mcH}dt=\int_0^T  \langle -i v_t+Lv,u \rangle_{\mcHp,\mcH}dt+\int_0^T \langle v,\psi h \rangle_{\mcHp,\mcH}dt.
\ee
By \eqref{eqdual}, being $u(\cdot,T)=0$ and $h(\cdot,t)=\Lambda^{-1}(\psi v(\cdot,t))$
\be
\langle v_0,-iu(x,0) \rangle_{\mcHp,\mcH}=\int_0^T \langle \psi v, \Lambda^{-1}(\psi v) \rangle_{\mcHp,\mcH} dt,
\ee
and therefore
\be
\langle v_0,S(v_0) \rangle_{\mcHp,\mcH}=\int_0^T \| \psi v \|^2_{\mcHp} dt \geq C \|v_0\|^2_{\mcHp}.
\ee
It follows from Lax Milgram that $S$ is an isomorphism.
\end{proof}

\subsection{Non linear system}

We are now in a position to present the local controllability of the non linear problem
\begin{align}
iu_t(x,t)&=Lu+m(u) u+\psi(x) h(x,t) \label{ec_nolineal}\\
u(x,0)&=u_0(x) \quad x\in \mathbb{R} \label{datoinicialnolineal}
\end{align}
which, as in the linear case, means the existence of a {\em
control} $h\in L^2(0,T,\mcH)$ such that the related solution
satisfies $u(x,T)=u_T(x).$

\begin{theorem} \label{local_control}\medskip

Let $T>0$ be fixed, then there exists $R>0$ such
that for every $u_0, u_T\in \mcH$ with
$\max\{\|u_0\|_{\mcH};\|u_T\|_{\mcH}\}<R$ there exists $h\in L^2(0,T;\mcH)$ such that the  unique
solution of (\ref{ec_nolineal})--(\ref{datoinicialnolineal}) satisfies
$u(x,T)=u_T(x).$

\end{theorem}

Equation \eqref{ec_nolineal}-\eqref{datoinicialnolineal} can be written in its integral form
\begin{eqnarray*}
u(x,t)=e^{-iLt}u_0(x)-i\int_{0}^t e^{iL(s-t)}m(u(x,s)) u(x,s) ds
-i \int_{0}^t e^{iL(s-t)}\psi(x) h(x,s)ds.
\end{eqnarray*}
We then set, for $v \in C(0,T,\mcH),$ the mapping that defines the nonlinear term
\begin{align}
 \mathcal{N}(v,0,t)&:=-i\int_{0}^t e^{iL(s-t)}\big(m(v(s))v(s)\big)
 ds.& \label{nonlinear_term} 
\end{align}

We next define $\Gamma:C(0,T,\mcH) \to C(0,T,\mcH)$ as
follows:

Given $v\in C(0,T,\mcH)$, we compute $\mathcal{N}(v,0,t)$ as in
\eqref{nonlinear_term}. Given the initial state $u_0$ and the
target state $u_T-\mathcal{N}(v,0,T)$, from Theorem
\ref{controlabilidadlineal} there exists a unique $h^{lin}\in
L^2(0,T,\mcH)$ such that the solution $\tilde{w}$ of the linear
equation \eqref{linealhomogeneo}-\eqref{datoinicial} with
$h=h^{lin}$
\be \label{wtilde} \tilde{w}(t)=e^{-iLt}u_0(x)-i \int_{0}^t
e^{iL(s-t)}\psi(x) h^{lin}(x,s)ds. \ee satisfies $\tilde{w} \in
C(0,T,\mcH)$ and \be \tilde{w}(T)=u_T-\mathcal{N}(v,0,T). \ee
Observe that $h^{lin}$  depends on $v$ and therefore $\tilde{w}$
also depends on $v$.

Let
\begin{align}\label{gama}
\Gamma(v)(t)&:=e^{-iLt}u_0+ \mathcal{N}(v,0,T)-
i \int_{0}^t e^{iL(s-t)}\psi(x) h^{lin}(x,s)ds.
\end{align}
Since $\tilde{w}\in C(0,T,\mcH)$ and is a solution of the linear equation \eqref{linealhomogeneo}, then $\Gamma(v)$ reads
\begin{align}\label{gama2}
\Gamma(v)(t)&:=\tilde{w}(t)+ \mathcal{N}(v,0,t)
\end{align}
and therefore $\Gamma(v) \in C(0,T,\mcH), \,\Gamma(v)(0)=u_0,$ and
$\Gamma(v)(T)=u_T$. We shall remark that any fixed point of
$\Gamma$ yields the function needed to build the control $h\in
L^2(0,T,\mcH).$ Hence, it only remains to show that $\Gamma$ has a
fixed point. Let $\delta
>0$ and set $K_\delta:=\{v \in C(0,T,\mcH): v(0)=u_0,\,v(T)=u_T,\,
\|v\|_{L^\8(t_0,T,\mcH)}<\delta \}$. As usual, we
must show that $K_\delta$ is left invariant by $\Gamma,$ and also
that this is a contractive mapping. With this in mind we list
below some useful estimates.

\begin{lemma} \label{fixed_point_estimate}
Let $R>0$ and let $u_0, u_T \in \mcH$ be such that
$\max\{\|u_0\|_{\mcH};\|u_T\|_{\mcH}\}<R,$ let also
$\delta>0$ and take $v,u\in K_\delta$.
Thus the following estimates hold,

\begin{itemize}
\item 
$\|\Gamma(v)\|_{L^\8(0,T,\mcH)}\leq A R + B \delta^3  $

\item $\|\Gamma(v)-\Gamma(u)\|_{L^\8(0,T,\mcH)}\leq C \delta^2 \|u-v\|_{L^\8(0,T,\mcH)}. $

\end{itemize}
where $A,B,C$ are positive constants.
\end{lemma}

\begin{proof}

These estimates follow from identities \eqref{nonlinear_term},
\eqref{wtilde}, \eqref{gama2} and inequalities
\eqref{semigroup-estimate}-\eqref{estimate_non_linear}:
\begin{align*}
\|\mathcal{N}(v,0,t)\|_\mcH
&\leq \int_{0}^{t} \|e^{iL(s-t)}\big(m(v(s))v(s)\big)\|_{\mcH} \,ds\\
&\leq \int_{0}^{t} \|m(v(s))v(s)\|_{\mcH}\,\big( 1+ (t-s)
\|\mu_{x} -\alpha_{x}\|_{L^\infty} \big) ds\\
&\leq \frac32 \left( 1+ T \|\mu_{x}-\alpha_{x}\|_{L^\infty}
\right) \int_{0}^{t} \|v(s)\|^3_{\mcH}\, ds\\
&\leq B({T,\|\mu_{x} -\alpha_{x}\|_{L^\infty}}) \|v\|^3_{L^\8(0,T,\mcH)}\\
& \leq B({T,\|\mu_{x} -\alpha_{x}\|_{L^\infty}}) \delta^3,
\end{align*}
and
\begin{align*}
\left\|\int_{0}^t e^{iL(s-t)}\psi(x) h^{lin}(x,s)ds\right\|_\mcH&\leq \int_{0}^{t} \|\psi h^{lin}(\cdot,s)\|_\mcH \big( 1+ (t-s)
\|\mu_{x} -\alpha_{x}\|_{L^\infty} \big)  ds \\
&\leq C_\psi \|h^{lin}\|_{L^2(0,T,\mcH)} \, \|1+ (t-s) \|\mu_{x} -\alpha_{x}\|_{L^\infty}\|_{L^2(0,t)} \\
&\leq  C({\psi,T,\|\mu_{x} -\alpha_{x}\|_{L^\infty}}) \left( \|u_0\|_\mcH + \|u_T-\mathcal{N}(v,0,T)\|_\mcH \right)\\
&\leq  C({\psi,T,\|\mu_{x} -\alpha_{x}\|_{L^\infty}}) \left( \|u_0\|_\mcH + \|u_T\|_\mcH + \|v\|_{L^\infty(0,T,\mcH)}^3\right) \\
&\leq A_1(\psi,T,\|\mu_{x} -\alpha_{x}\|_{L^\infty})
R+B(T,\|\mu_{x} -\alpha_{x}\|_{L^\infty}) \delta^3
\end{align*}

For the second assertion note that \be \label{contraction}
\Gamma(v)(t)-\Gamma(u)(t)=-i\int_{0}^t e^{iL(s-t)}
\big(m(v)(s)v(s)-m(u)(s)u(s)\big) ds. \ee

A similar reasoning leads us to the inequality
\begin{align*}
\|\Gamma(v)(t)-\Gamma(u)(t)\|_{\mcH}
&\leq  \int_{0}^{t} \|e^{iL(s-t)}\big(m(v(s))v(s)-m(u(s))u(s)\big)\|_{\mcH} \,ds\\
&\leq  B({T,\|\mu_{x} -\alpha_{x}\|_{L^\infty}})
\int_{0}^{t}
\|m(v(s))v(s)-m(u(s))u(s)\|_{\mcH}\, ds\\
&\leq B({T,\|\mu_{x} -\alpha_{x}\|_{L^\infty}}) \int_{0}^{t}
\left(\|v\|^2_{\mcH}+\|v\|_{\mcH}\|u\|_{\mcH}+\|u\|^2_{\mcH}
\right)\|v-u\|_{\mcH}\,  ds \\
&\leq B({T,\|\mu_{x} -\alpha_{x}\|_{L^\infty}}) \delta^2
\|v-u\|_{L^\infty(0,T,\mcH)},
\end{align*}

from where second estimate follows easily. This finishes the proof.

\end{proof}

\begin{proof}[Proof of Theorem \ref{local_control}]

As we state above, it relies on a fixed point argument. Set
$K_\delta:=\{v \in C(0,T,\mcH): v(0)=u_0,\,v(T)=u_T,\,
\|v\|_{L^\8(0,T,\mcH)}<\delta \}$. Using the estimates given by Lemma
\ref{fixed_point_estimate}, we get the following sufficient
conditions
\begin{align*}
A R + B \delta^3 \leq \delta\\
C \delta^2 <1 \end{align*}
which are easily satisfied taking $\delta=2RA$ and
$R<\min\{\frac{1}{2\sqrt{C}A},\frac{1}{2\sqrt{2B}A} \}$.
\end{proof}

\section{Non controllability for compactly supported controls}

Throughout this section we shall focus our attention to controls
$\psi(x) h(x,t)$ with $\mathrm{Supp}(\psi)$ compact, and consider
two different situations, depending on the linear term: $L_\mu=-\partial_x^2+\mu,$ which has a discrete
spectrum, and $L_e=-\partial_x^2-x$ with a continuous spectrum.
 The negative result concerning the related
exact controllability for the linear problem is similar to the one given in \cite{ILT2}, however our problem is posed in $\mcH$ which is not $L^2$ but a suitable Sobolev space. For this reason we shall adapt both the result and its proof, and this heavily relies upon the spectral properties reported in section 2. Actually, since the proof relies on a special feature of the eigenstates of the linear operator, we shall use the unitary group $U_+$, and the eigenfunctions $\{\vi_N\}_{N\in \N}$ of the auxiliar operator $L_+:=-\partial_x^2+|x|$ yielded by Lemma \ref{L+operator}.

\subsection{Discrete spectrum}

We first consider the non--controllability result for the model
equation,
\begin{align} \label{no_controla_bis}
  &iu_t(x,t)=L_\mu u(x,t)+\psi(x) h(x,t),\quad x\in \mathbb{R}, \\
&u(x,0)=u_0(x), \qquad u(x,T)=u_T(x),
\end{align}
with $\mathrm{Supp}(\psi)$ compact. The main result reads as
follows.

\begin{theorem}
The exact internal distributed control is not possible, i.e. for a
given target state $u_T\in \mcH$ there exist a bounded open set
$\Omega \subseteq \R$ and an initial function $u_0$ such that
there is no control function $h$ and no constant $C=C(\Omega,T)>0$
such that the equation \eqref{no_controla_bis} holds with
$u(0)=u_0,$ $u(T)=u_T,$ and
  $\|h\|_{L^1(0,T,\mcH)}\leq C\left(\|u_0\|_{\mcH}+\|u_T\|_{\mcH}\right)$
\end{theorem}

\begin{proof}
As in \cite{ILT2} we argue by contradiction. Let $\Omega$ be a
fixed finite interval and take $\vi_N,$ the $N$--th eigenfunction of $L_+,$ as a
target state, and assume that there exist a time $T>0,$ a control
function $h_N \in L^2(0,T,\mcH),$ a constant $C(\Omega,T),$ with $\|h\|_{L^2(0,T,\mcH)}\leq C(\Omega,T) (\|u_0\|_{\mcH}+\|\vi_N\|_\mcH)$
an initial state $u_0$ and a solution $u_N$ of
\eqref{no_controla_bis}. Let $U_+(t)$ be the unitary group
generated by $L_+$ in $\mcH,$ since $L_\mu=L_+ +b$ where $b(x)=\mu(x)-|x|$ has compact support,  from Duhamel identity we have:
$$\vi_N(x)=U_+(T)u_0(x)-i\int_0^T U_+(T-s) (\psi h_N+b u_N) ds.$$
Since $U_+(t)\phi=\sum e^{-it\lambda_N}\widehat{\phi}(N)\vi_N(x),$
where $\widehat{\phi}(N)=\int\phi(x)\vi_N(x)dx$ are the related
Fourier coefficients, after taking the $L^2$-inner product with
$\vi_N$ we get
\begin{align}\label{proof_contra}
1=e^{-iT\lambda_N}\langle u_0; \vi_N \rangle-i\int_0^T e^{-i(T-s)\lambda_N}\langle \psi h_N+b u_N; \vi_N \rangle ds.
\end{align}

Since $u_0\in \mcH$ the first term goes to zero. The second term verifies
\begin{align*}
 \langle\psi h_N; \vi_N \rangle&=\lambda_N^{-1}\langle \psi h_N+b u_N; L_+\vi_N \rangle\\
&=\lambda_N^{-1}\langle L_+(\psi h_N+b u_N);\vi_N \rangle \\
&=\lambda_N^{-1}\langle (\mu \psi+\psi_{xx})h_N+\psi_{x} (h_N)_{x};\vi_N\rangle -\lambda_N^{-1}\langle \psi (h_N)_{x};(\vi_N)_{x} \rangle\\
&\quad +\lambda_N^{-1}\langle (\mu b+b_{xx})u_N+b_{x} (u_N)_{x};\vi_N\rangle -\lambda_N^{-1}\langle b (u_N)_{x};(\vi_N)_{x} \rangle
\end{align*}

From Lemma
\ref{L+operator} we see that the eigenfunctions
$\{\vi_N\}_{N\in \N}$ satisfy $\|\vi_N\|_{L^2}=1,$ $\|\vi_N\|_{\mcH}=\lambda_N^{1/2},$ and $\lambda_N^{-1}\|(\vi_N)_{x}\|_{L^2(\Omega)}\sim \lambda_N^{-3/4},$ we also recall that both $\psi$ and $b$ have compact support. With this in mind we get:
\begin{align*}
 \left|i\int_0^T e^{-i(T-s)\lambda_N}\langle \psi h_N+b u_N; \vi_N \rangle \right|&\leq \int_0^T \left| \langle\psi h_N+b u_N; \vi_N \rangle\right|\\
 &\leq C(\psi,b) \lambda_N^{-1} \|\vi_N\|_{L^2} \int_0^T (\|u_N\|_\mcH+\|h_N\|_\mcH)\\
 &\quad +\lambda_N^{-1}\|(\vi_N)_{x}\|_{L^2(\Omega)} \int_0^T  (\|u_N\|_\mcH+\|h_N\|_\mcH)\\
&\leq \lambda_N^{-1} C(\psi,\Omega,T) (\|u_0\|_\mcH+\lambda_N^{1/2})(1+\lambda_N^{1/4})
\end{align*}
which goes to zero as $N$ goes to infinity. This contradicts identity \eqref{proof_contra}, and this finishes the proof.
\end{proof}

\subsection{Continuous spectrum}

We now consider the non--controllability result for the linear model
equation, with $L_e=-\partial_x^2-x,$
\begin{align} \label{no_controla}
  &iu_t(x,t)=L_e u(x,t)+\psi(x) h(x,t),\quad x\in \mathbb{R}, \\
&u(x,0)=u_0(x), \qquad u(x,T)=u_T(x),
\end{align}
with $\mathrm{Supp}(\psi)$ compact. The main result reads as
follows.

\begin{theorem}\label{nohaycontrol}
The exact internal distributed control is not possible, i.e. for a
given target state $u_T\in \mcH$ there exist a bounded open set
$\Omega \subseteq \R$ and an initial function $u_0$ such that
there is no control function $h$ and no constant $C=C(\Omega,T)>0$
such that the equation \eqref{no_controla} holds with $u(0)=u_0,$
$u(T)=u_T,$ and
  $\|h\|_{L^1(0,T,\mcH)}\leq C\left(\|u_0\|_{\mcH}+\|u_T\|_{\mcH}\right)$
\end{theorem}

\begin{remark}
 As for the result of the previous subsection we follow the ideas
  of Theorem 3 of \cite{ILT2}, but in order to accomplish the task we need an extra ingredient given by the following Lemma.
\end{remark}

\begin{lemma}\label{commutator}
 Let $U_e(t)$ be the group generated by $L_e:=-\partial^2_x+x.$ Then $[\partial_x:U_e(r)]=rU_e(r)$
\end{lemma}

\begin{proof}
 We start noting that $[\partial_x:L_e]=[\partial_x:x]=1$ and $[\partial_x:L_e^{M+1}]=[\partial_x:L_e^{M}]L+L^M [\partial_x:L_e].$ An inductive argument shows the identity $[\partial_x:L_e^{M+1}]=(M+1)L^M.$ For $\phi$ in the Schwarz space we have
\begin{align*}
 [\partial_x:U_e(t)]\phi&=\sum_{M\geq 0} \frac{t^M}{M!}[\partial_x:L_e^{M}]\phi\\
&=\sum_{M\geq 0} \frac{t^{M+1}}{(M+1)!}[\partial_x:L_e^{M+1}]\phi\\
&=t U_e(t)\phi
\end{align*}
A density argument allows us to extend the result for $\phi \in \mcH.$
\end{proof}

\begin{proof}[Proof of Theorem \ref{nohaycontrol}]
We first set $\Psi(x)\in C_0^\infty(\R)$ such that $0\leq\Psi(x),$ $\mathrm{Supp}(\Psi)=[-1;1],$ and
$1=\int\Psi(x),$ and take
$\Psi_\varepsilon=\varepsilon^{-1}\Psi(\varepsilon^{-1}x).$ We below collect the behavior of the different norms
involved in the proof, their validity is evident and will not be reported.
\begin{align}\label{different_norms}
\|\Psi_\varepsilon\|_{L^1}&=\|\Psi\|_{L^1}=1\\
\|\Psi_\varepsilon\|_{L^2}&=\varepsilon^{-1/2}\|\Psi\|_{L^2}\\
\|\Psi_\varepsilon\|_{L^2_\mu}&\leq \varepsilon^{-1/2}(1+\varepsilon)^{1/2}\|\Psi\|_{L^2}\\
\|(\Psi_\varepsilon)_{x}\|_{L^1}&=\varepsilon^{-1}\|\Psi_{x}\|_{L^1}\\
\|(\Psi_\varepsilon)_{x}\|_{L^2}&=\varepsilon^{-3/2}\|\Psi_{x}\|_{L^2}
\end{align}

  We also add, for a fixed $T>0,$ the function
  $\phi_\varepsilon:=U_e(2T)\Psi_\varepsilon,$ where $U_e$ is the
  related unitary group, and notice that
  $\|\phi_\varepsilon\|_{L^2}^2=\|\Psi_\varepsilon\|_{L^2}^2=\varepsilon^{-1}\|\Psi\|^2_{L^2}.$
  We now argue by contradiction. Assume the exact controllability of \eqref{no_controla},
  then there exist $h_\varepsilon \in L^2(0,T,\mcH)$ such that
  \begin{eqnarray*}
    \|h_\varepsilon\|_{L^2(0,T,\mcH)}\leq C(\|u_0\|_{\mcH},
    \|\phi_\varepsilon\|_{\mcH}),
  \end{eqnarray*}
and a solution $u_\varepsilon(x,t)$ of \eqref{no_controla} with
$u_\varepsilon(\cdot,T)=\phi_\varepsilon,$ and $u_0\in \mcH$
arbitrary.

From Duhamel identity we have
$$\phi_\varepsilon=U_e(T)u_0-i\int_0^T U_e(T-s) (\psi
h_\varepsilon)ds$$ and taking the $L^2$ inner--product with $L_+ \phi_\varepsilon$ we get for the left hand side:
$$ \langle \phi_\varepsilon ;L_+ \phi_\varepsilon \rangle=\langle \Psi_\varepsilon; U_e(-2T)(-\partial_x^2)U_e(2T) \Psi_\varepsilon\rangle$$
Before going further we develop an useful identity, based on the commutator relation given by Lemma \ref{commutator}:
\begin{align}\label{exp_nabla_exp}
\nonumber U_e(r)(-\partial_x^2)U_e(s)&=-[U_e(r): \partial_x]\partial_x U_e(s)-\partial_x U_e(r)\partial_x U_e(s)\\ \nonumber
&=r U_e(r)\partial_xU_e(s)-s\partial_x U_e(r)U_e(s)-\partial_x U_e(r)U_e(s)\partial_x \\
&=-r^2U_e(r+s)+(r-s)\partial_x U_e(r+s)-\partial_x U_e(r+s)\partial_x
\end{align}
With this result, the left hand side reads
\begin{align*}
 \langle \phi_\varepsilon ;L_+ \phi_\varepsilon \rangle&=-4T^2\|\Psi_\varepsilon\|^2_{L^2}-4T\langle \Psi_\varepsilon;\partial_x\Psi_\varepsilon \rangle-\langle \Psi_\varepsilon;\partial_x^2\Psi_\varepsilon \rangle+\langle \phi_\varepsilon;\mu\phi_\varepsilon \rangle\\
&=-4T^2\varepsilon^{-1}\|\Psi\|^2_{L^2}+\varepsilon^{-3}\|\Psi_x\|^2_{L^2}+\|\phi_\varepsilon\|^2_{L^2_\mu}\\
\end{align*}
The last term is bounded with the help of Lemma \ref{semigroup-estimate}
$$\|\phi_\varepsilon\|^2_{L^2_\mu}\leq \|\Psi_\varepsilon\|^2_{L^2_\mu}+4T\|\Psi_\varepsilon\|_{L^2}\|(\Psi_\varepsilon)_{x}\|_{L^2}+ 4T^2 \|\Psi_\varepsilon\|^2_{L^2}$$
revious estimates altogether yield for the left hand side:
\begin{align}\label{left_side}
\varepsilon^3  \langle \phi_\varepsilon ;L_+ \phi_\varepsilon \rangle=\|\Psi_x\|^2_{L^2}+O(\varepsilon).
\end{align}

We now turn to the right hand side, after multiplying by $\varepsilon^3,$
\begin{align*}
\varepsilon^3 \langle
U_e(T)u_0;L_+U_e(2T)\Psi_\varepsilon\rangle-i\varepsilon^3\int_0^T
\langle U_e(T-s) (\psi h)ds;L_+U_e(2T)\Psi_\varepsilon\rangle.
\end{align*}
The first term goes to zero as easily follows from Lemma \ref{semigroup-estimate} and estimates \eqref{different_norms}:
\begin{align*}
 \varepsilon^3\left|\langle
U_e(T)u_0;L_+U_e(2T)\Psi_\varepsilon\rangle\right|& \leq \|\phi_\varepsilon\|_\mcH \|U_e(T)u_0\|_\mcH\\
&\leq C(T)\|\Psi_\varepsilon\|_\mcH \|u_0\|_\mcH\\
&\leq C(T,u_0,\Psi) \varepsilon^{3/2}.
\end{align*}

The second term is splitted as
\begin{align*}
-i\varepsilon^3\int_0^T
\langle U_e(T-s) (\psi h)ds;(-\partial_x^2)U_e(2T)\Psi_\varepsilon\rangle-i\varepsilon^3\int_0^T
\langle U_e(T-s) (\psi h)ds;\mu U_e(2T)\Psi_\varepsilon\rangle.
\end{align*}
and each term is treated separately. For the later we apply a similar procedure as for the initial datum:
\begin{align*}
 \varepsilon^3\left|\langle
U_e(T-s) (\psi h)ds;\mu U_e(2T)\Psi_\varepsilon\right|& \leq \|\phi_\varepsilon\|_\mcH \|U_e(T-s)\psi h_\varepsilon\|_{L^2_\mu}\\
&\leq C(T,\Psi)\varepsilon^2\|\psi h_\varepsilon\|^{1/2}_{L^2} \|(\psi h_\varepsilon)_x\|^{1/2}_{L^2}\\
&\leq C(T,\Psi,\Omega)\varepsilon^2\|h_\varepsilon\|_{H^1}\\
&\leq C(T,\Psi,\Omega)\varepsilon^{1/2}
\end{align*}
where the former is handled using the $L^1-L^\infty$ estimate displayed in Corollary
\ref{L1_L8_estimate}. To see this we first apply the identity \eqref{exp_nabla_exp} and get:
$$U_e(-2T)(-\partial_x^2)U_e(T-s)=-4T^2U_e(-T-s)+(s-3T)\partial_x U_e(-T-s)-\partial_x U_e(-T-s)\partial_x.$$
This leads to:
\begin{align*}
\left|\langle U_e(T-s) (\psi h)ds;(-\partial_x^2)U_e(2T)\Psi_\varepsilon\rangle\right| &\leq 4T^2 \|\Psi_\varepsilon\|_{L^1} \|U_e(-T-s)(\psi h_\varepsilon)\|_{L^\infty}\\&\quad +3T \|(\Psi_\varepsilon)_x\|_{L^1} \|U_e(-T-s)(\psi h_\varepsilon)\|_{L^\infty}\\
&\quad + \|(\Psi_\varepsilon)_x\|_{L^1} \|U_e(-T-s)(\psi h_\varepsilon)_x\|_{L^\infty}\\
&\leq C(\Omega,T) \|h_\varepsilon\|_{L^2}+C(\Omega,T,\Psi) \varepsilon^{-1} \|h_\varepsilon\|_{H^1}\\
&\leq C(\Omega,T,\Psi,u_0) \varepsilon^{-5/2}
\end{align*}
where we have used the estimates $\|\psi h_\varepsilon\|_{L^1}\leq C(\Omega,T)\|h_\varepsilon\|_{L^2},$ $\|(\psi h_\varepsilon)_x\|_{L^1}\leq C(\Omega,T)\|h_\varepsilon\|_{H^1},$ and the fact that $|-T-s|^{-1/2}\leq T^{-1/2}.$
Integrating in $[0,T]$ and multiplying by $\varepsilon^3$ we see that the right hand side tends to zero, contradicting the estimate \eqref{left_side}. This finishes the proof.

\end{proof}

\end{document}